\documentclass{article}

\usepackage{amsthm}
\newtheorem{theorem}{Theorem}
\newtheorem{lemma}{Lemma}
\newtheorem{corollary}{Corollary}

\usepackage{mathtools}
\DeclareMathOperator{\NE}{NE}
\DeclareMathOperator{\NW}{NW}
\DeclareMathOperator{\SE}{SE}
\DeclareMathOperator{\SW}{SW}

\usepackage{microtype}
\usepackage{xcolor}

\usepackage{slantsc}
\AtBeginDocument{%
    \DeclareFontShape{OT1}{cmr}{m}{scit}{<->ssub*lmr/m/scsl}{}%
}

\makeatletter
\g@addto@macro\bfseries{\boldmath}
\makeatother

\binoppenalty = \maxdimen
\relpenalty = \maxdimen

\begin{document}

\title{Drawing Two Posets}
\author{Guido Br\"uckner \and Vera Chekan}
\date{}

\maketitle

\begin{abstract}
    We investigate the problem of drawing two posets of the same ground set so that one is drawn from left to right and the other one is drawn from the bottom up.
    The input to this problem is a directed graph~$G = (V, E)$ and two sets~$X, Y$ with~$X \cup Y = E$, each of which can be interpreted as a partial order of~$V$.
    The task is to find a planar drawing of~$G$ such that each directed edge in~$X$ is drawn as an~$x$-monotone edge, and each directed edge in~$Y$ is drawn as a~$y$-monotone edge.
    Such a drawing is called an \emph{$xy$-planar} drawing.

    Testing whether a graph admits an~$xy$-planar drawing is~\textsf{NP}-complete in general.
    We consider the case that the planar embedding of~$G$ is fixed
    and the subgraph of~$G$ induced by the edges in~$Y$ is a connected spanning subgraph of~$G$ whose upward embedding is fixed.
    For this case we present a linear-time algorithm that determines whether~$G$ admits an~$xy$-planar drawing and, if so, produces an~$xy$-planar polyline drawing with at most three bends per edge.
\end{abstract}

\section{Introduction}
\label{sec:introduction}

A partial order~$<$ over a set~$V$ can be interpreted as a directed graph~$G$ by interpreting the elements of~$V$ as the vertices of~$G$, and interpreting the fact~$u < v$ for~$u, v \in V$ as a directed edge from~$u$ to~$v$ in~$G$.
In an upward drawing of~$G$, every directed edge is drawn as an increasing~$y$-monotone curve.
Such a drawing is a bottom-up visualization of the partial order~$<$.
If the drawing is also planar, then it is especially easy to understand for humans~\cite{DBLP:conf/gd/Purchase97}.
Testing graphs for upward planarity is \textsf{NP}-hard in general~\cite{DBLP:journals/siamcomp/GargT01}, but feasible in linear time for graphs with a single source~\cite{DBLP:journals/siamcomp/BertolazziBMT98} and for graphs with a fixed underlying planar embedding~\cite{DBLP:journals/algorithmica/BertolazziBLM94}.

More recent research has sought to extend this concept to two directions.
The input to \textsc{Bi-Monotonicity} is an undirected graph whose vertices have fixed coordinates and the task is to draw each edge as a curve that is both~$x$-monotone and~$y$-monotone while maintaining planarity.
This problem is \textsf{NP}-hard in general~\cite{DBLP:journals/talg/KlemzR19}.
In \textsc{Upward-Rightward Planarity}, the question is whether there exists a planar drawing of a directed graph in which each edge is~$x$-monotone or~$y$-monotone.
Every planar directed graph has an upward-rightward straight-line drawing in polynomial area that can be computed in linear time~\cite{DBLP:conf/iisa/GiacomoDKLM14}.
The input to \textsc{HV-Rectilinear Planarity} is an undirected graph~$G$ with vertex-degree at most four where each edge is labeled either as horizontal or as vertical.
The task is to find a planar drawing of~$G$ where each edge labeled as horizontal (vertical) is drawn as a horizontal (vertical) line segment.
This problem is \textsf{NP}-hard in general~\cite{DBLP:journals/jcss/DidimoLP19}.
A windrose graph consists of an undirected graph~$G$ and for each vertex~$v$ of~$G$ a partition of its neighbors into four sets that correspond to the four quadrants around~$v$.
A windrose drawing of~$G$ is a drawing such that for each vertex~$v$ of~$G$ each neighbor lies in the correct quadrant and all edges are represented by curves that are monotone with respect to each axis.
Testing graphs for windrose planarity is \textsf{NP}-hard in general, but there exists a polynomial-time algorithm for graphs with a fixed underlying planar embedding~\cite{DBLP:journals/talg/AngeliniLBDKRR18}.

In this paper, we investigate a new planarity variant called \emph{$xy$-planarity}.
Let~$G = (V, E)$ be a directed graph together with two sets~$X, Y$ with~$X \cup Y = E$.
In an \emph{$xy$-drawing} of~$G$ each directed edge in~$X$ is drawn as a strictly increasing~$x$-monotone curve and each directed edge in~$Y$ is drawn as a strictly increasing~$y$-monotone curve.
Hence, an~$xy$-drawing of~$G$ is a left-to-right visualization of the partial order defined by~$X$ and a bottom-up visualization of the partial order defined by~$Y$, i.e., it is a drawing of two posets on the same ground set~$V$.
A planar~$xy$-drawing is an \emph{$xy$-planar} drawing.
The study of~$xy$-planarity has been proposed by Angelini et al.~\cite{DBLP:journals/talg/AngeliniLBDKRR18}.
Because~$xy$-planarity generalizes both upward planarity and windrose planarity, we immediately obtain the following.

\begin{theorem}
    Testing graphs for~$xy$-planarity is \textsf{NP}-complete.
\end{theorem}

We therefore consider the restricted case where the planar embedding of~$G$ is fixed, and the~$Y$-induced subgraph of~$G$ is a connected spanning subgraph of~$G$ whose upward embedding is fixed.
For this case we present a linear-time algorithm that determines whether~$G$ admits an~$xy$-planar drawing.
Our algorithm uses several structural insights.
First, in Section~\ref{sec:windrose-planarity}, we provide a new, simple combinatorial characterization of windrose planar embeddings.
From each~$xy$-planar drawing of~$G$ we can derive an embedded windrose planar graph~$G^+$.
Using our combinatorial characterization of windrose planar embeddings we can simplify~$G^+$.
In Section~\ref{sec:windrose-drawings-derived-from-xy-drawings} we show that in this simplified graph, every edge of the original graph~$G$ corresponds to one of four windrose planar gadgets.
To test~$G$ for~$xy$-planarity, we show in Section~\ref{sec:xy-planarity-testing} how to determine in linear time whether there exists a choice of one gadget for each edge of~$G$ that leads to a windrose planar embedding.
In the positive case our algorithm outputs an~$xy$-planar drawing on a polynomial-size grid where each edge has at most three bends.
If every windrose planar graph has a straight-line drawing (an open question), then each edge has at most one bend, which we show to be optimal.

\section{Preliminaries}
\label{sec:preliminaries}

We use standard terminology concerning (upward) planar graph drawings and embeddings.
Let~$G = (V, E)$ be a connected graph.
A \emph{drawing} of~$G$ maps each vertex to a point in the plane and each edge to a finite polygonal chain between its two endpoints.
A drawing is \emph{planar} if distinct edges do not intersect except in common endpoints.
An \emph{embedding} of~$G$ consists of a counter-clockwise cyclic order of edges incident to each vertex of~$G$.
A planar drawing of~$G$ induces an embedding of~$G$.
An embedding~$\mathcal E$ of~$G$ is \emph{planar} if there exists a planar drawing of~$G$ that induces~$\mathcal E$.
For an inner face (the outer face)~$f$ of~$\mathcal E$, define the \emph{boundary} of~$f$ as the clockwise (counter-clockwise) cyclic walk on the edges incident to~$f$.

Let~$G$ be a directed graph.
A drawing of~$G$ is \emph{upward} if each directed edge~$(u, v)$ is drawn as a connected series of strictly increasing~$y$-monotone line segments.
A drawing is \emph{upward planar} if it is both upward and planar.
An \emph{upward embedding} of~$G$ consists of left-to-right orders of incoming and outgoing edges incident to each vertex of~$G$.
An upward planar drawing of~$G$ induces an upward embedding of~$G$.
An upward embedding of~$G$ is \emph{planar} if it is induced by an upward planar drawing of~$G$.
An upward embedding induces an \emph{underlying embedding} of~$G$ by concatenating the left-to-right order of incoming edges and the reversed left-to-right order of outgoing edges into one counter-clockwise cyclic order around each vertex.
An embedding of~$G$ is \emph{bimodal} if the incoming (outgoing) edges appear consecutively around each vertex.
The underlying embedding of an upward embedding is bimodal.
A vertex of~$G$ is a \emph{sink} (\emph{source}) if it is incident only to incoming (outgoing) edges.
A vertex that is neither a source nor a sink of~$G$ is called an \emph{inner vertex}.
Consider an upward planar drawing of~$G$ and its underlying planar embedding~$\mathcal E$ of~$G$.
For three consecutive vertices~$u, v, w$ on the boundary of a face~$f$, the vertex~$v$ is called a \emph{face-source} (\emph{face-sink}) of~$f$ if the edges~$uv$ and~$vw$ are both directed away from (towards)~$v$.
Define~$n_f$ as the number of face-sources of~$f$, which equals the number of face-sinks of~$f$.
A \emph{sink/source assignment}~$\psi: v \mapsto (e, e')$ maps each source and sink~$v$ to two edges~$e, e'$ incident to~$v$ so that~$e$ immediately precedes~$e'$ in counter-clockwise cyclic order of edges incident to~$v$ defined by~$\mathcal E$.
This corresponds to a unique face~$f$ of~$\mathcal E$ such that~$e$ immediately precedes~$e'$ on the boundary of~$f$.
Thus, we say that~$\psi$ assigns~$v$ to~$f$.
The assignment is \emph{upward consistent} if~$\psi$ assigns~$n_f + 1$ vertices to one face, and~$n_f - 1$ vertices to all other faces.
The face to which~$n_f + 1$ vertices are assigned is the outer face.
From an embedding~$\mathcal E$ and a sink/source assignment~$\psi$ an upward embedding of~$G$ can be obtained by splitting for each sink (source)~$v$ the counter-clockwise cyclic order of edges incident to~$v$ defined by~$\mathcal E$ between the two edges~$e, e'$ with~$\psi(v) = (e, e')$ to obtain the right-to-left (left-to-right) order of incoming (outgoing) edges.
In the reverse direction, from an upward embedding of~$G$ a sink/source assignment~$\psi$ can be obtained as follows.
Assign each sink~$v$ to~$(e, e')$ where~$e$ and~$e'$ are the rightmost and leftmost incoming edge incident to~$v$, respectively.
Assign each source~$v$ to~$(e, e')$ where~$e$ and~$e'$ are the leftmost and rightmost outgoing edge incident to~$v$, respectively.
Thus, upward embeddings are equivalent to planar embeddings together with a sink/source assignment.
The following was observed by Bertolazzi et al.~\cite{DBLP:journals/algorithmica/BertolazziBLM94} for biconnected graphs, we provide a straight-forward extension to simply-connected graphs.

\begin{lemma}
    \label{lem:upward-planarity}
    Let~$G$ be a connected directed acyclic graph together with an embedding~$\mathcal E$.
    There exists an upward planar embedding of~$G$ whose underlying embedding is~$\mathcal E$ if and only if~$\mathcal E$ is planar and bimodal, and it admits an upward consistent assignment.
\end{lemma}

\begin{proof}
    Bertolazzi et al.~have proven the statement for biconnected graphs~\cite{DBLP:journals/algorithmica/BertolazziBLM94}.
    We extend it to simply connected graphs by induction over the number of maximal biconnected components of~$G$.
    If there is one such component the result of Bertolazzi et al.~applies.
    If~$G$ consists of a single vertex or a single edge, the statement holds trivially.
    Let~$v$ be a cutvertex of~$G$.
    Then there exist two edges~$e_1 = \{v, w_1\}$ and~$e_2  = \{v, w_2\}$ incident to~$v$ such that~$e_1$ immediately follows~$e_2$ in the counter-clockwise order of edges around~$v$ and~$w_1$ and~$w_2$ belong to different maximal biconnected components of~$G$.

    Let~$\Gamma$ be an upward planar drawing of~$G$ with underlying embedding~$\mathcal E$.
    Let~$f$ be a face such that~$e_1$ and~$e_2$ appear consecutively on its boundary.
    Moving closely along~$e_1$ and~$e_2$ in~$f$, insert a new vertex~$x$ and connect it to~$w_1$ and~$w_2$ as shown in Figure~\ref{fig:weakly-connected-augmentation}.
    \begin{figure}
        \centering
        \includegraphics{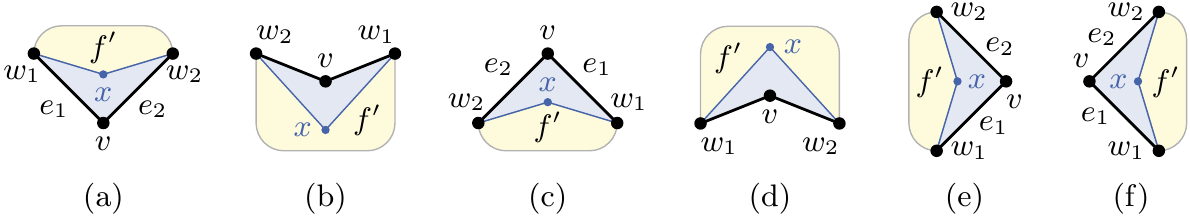}
        \caption{
           Reducing the size of the block-cutvertex tree by augmenting the graph given a cut-vertex~$v$ if~$v$ is (a) a sink, (b) a source, or (c) an inner vertex.
        }
        \label{fig:weakly-connected-augmentation}
    \end{figure}
    Note that this case distinction can be made on the basis of~$\Gamma$, although it could not be made solely on the basis of~$\mathcal E$.
    This separates~$f$ into one face of size four (shaded in blue), and another face~$f'$ (shaded in yellow).
    Let~$G'$ denote the graph obtained in this way, and let~$\Gamma'$ denote the drawing obtained in this way.
    In~$G'$ the vertices~$w_1$ and~$w_2$ belong to the same maximal biconnected component.
    Thus,~$G'$ has at least one less maximal biconnected component than~$G$.
    By induction, the underlying planar embedding~$\mathcal E'$ of~$\Gamma'$ then admits an upward consistent assignment~$\psi'$.
    By construction of the gadgets the number of face-sources and face-sinks of~$f$ equals the number of face-sources and face-sinks of~$f'$, respectively.
    To obtain an upward consistent assignment~$\psi$ of~$\mathcal E$ define~$\psi: v \mapsto (e_2, e_1)$ if and only if~$\psi': v \mapsto (\{x, w_2 \}, \{ x, w_1 \})$.

    Now assume that~$\mathcal E$ admits an upward consistent assignment~$\psi$.
    Repeat a similar argument as above.
    This time, the case distinction is made not based on~$\Gamma$, but based on~$\psi$, obtaining an upward consistent assignment~$\psi'$ of~$\mathcal E'$.
    By induction there exists an upward planar embedding~$\mathcal U'$ of~$G'$ whose underlying embedding is~$\mathcal E'$.
    Removing~$x$ from~$\mathcal U'$ gives an upward planar embedding~$\mathcal U$ of~$G$ whose underlying embedding is~$\mathcal E$.
\end{proof}

\section{Combinatorial View of Windrose Planarity}
\label{sec:windrose-planarity}

A \emph{windrose graph} is a directed graph~$G$ whose edges are labeled as either north-west (NW) or north-east (NE).
A \emph{windrose drawing} of~$G$ is an upward drawing of~$G$ where all NW (NE) edges decrease (increase) monotonically along the~$x$-axis.
In this way, the neighbors of each vertex are partitioned into the four quadrants of the plane around~$v$.
An upward planar embedding of~$G$ is \emph{windrose planar} if it is induced by a windrose planar drawing of~$G$.
Let~$\mathcal U$ be an upward planar embedding of~$G$ and let~$v$ be a vertex of~$G$.
We say~$\mathcal U$ is \emph{windrose consistent at~$v$} if
(i) the NW edges precede the NE edges in the left-to-right order of outgoing edges incident to~$v$, and
(ii) the NE edges precede the NW edges in the left-to-right order of incoming edges incident to~$v$.
We say that~$\mathcal U$ is \emph{windrose consistent} if it is windrose consistent at all vertices of~$G$ and show the following.

\begin{lemma}
    \label{lem:windrose-planar-embedding-consistent}
    Let~$G$ be a directed graph together with an upward planar embedding~$\mathcal U$.
    Then~$\mathcal U$ is a windrose planar embedding of~$G$ if and only if~$\mathcal U$ is windrose consistent.
\end{lemma}

\begin{proof}
    If~$\mathcal U$ is a windrose planar embedding then~$\mathcal U$ is windrose consistent.
    For the reverse direction, assume that~$\mathcal U$ is windrose consistent.
    We show that~$\mathcal U$ is windrose planar by induction over the number of NW edges in~$G$.
    Every upward planar graph admits a straight-line drawing~\cite{DBLP:journals/tcs/BattistaT88}.
    If~$G$ contains no NW edge, first vertically stretch such a straight-line drawing so that all slopes lie in the interval~$(\pi/4, 3\pi/4)$, then rotate the stretched drawing so that all slopes lie in the interval~$(0, \pi/2)$ to obtain a windrose planar drawing of~$G$.
    For the inductive step, the idea is to find an NW edge in~$G$ that can be relabeled as NE, where the result is a graph~$G'$ such that~$\mathcal U$ is a windrose consistent embedding of~$G'$.
    To this end, construct an \emph{edge dependency graph}~$D = (E, A)$ as follows.
    The edges of~$G$ are the nodes of~$D$.
    Construct the directed arcs of~$D$ as follows.
    For each vertex~$v$ of~$G$,
    consider the outgoing edges~$e^+_1, e^+_2, \dots, e^+_m$ in the left-to-right order prescribed by~$\mathcal U$.
    For~$1 \le i < m$ add the arc~$(e^+_i, e^+_{i + 1})$ to~$D$.
    Moreover, consider the incoming edges~$e_1^-, e_2^-, \dots, e_n^-$ in the left-to-right order prescribed by~$\mathcal U$.
    For~$1 \le i < n$ add the arc~$(e^-_{i + 1}, e^-_i)$ to~$D$.
    Consider an upward planar straight-line drawing~$\Gamma$ of~$G$.
    An arc~$(e, f)$ of the dependency graph~$D$ can be interpreted as the edge~$e$ having a greater slope than the edge~$f$ in~$\Gamma$.
    This directly implies that~$D$ is acyclic.
    The edge dependency graph~$D$ contains no arc from an NE edge to an NW edge by construction.
    Thus, if~$G$ contains an NW edge, there exists an NW edge~$e$ whose predecessors are all NW edges and whose successors in~$D$ are all NE edges.
    Let~$G'$ denote the graph obtained from~$G$ by relabeling~$e$ as NE.
    Observe that~$\mathcal U$ is a windrose consistent upward planar embedding of~$G'$.
    We show how to transform a windrose planar drawing~$\Gamma'$ of~$G'$ whose upward embedding is~$\mathcal U$ into a windrose planar drawing~$\Gamma$ of~$G$ whose upward embedding is~$\mathcal U$.
    We may assume that the vertices are in general position.
    Otherwise, this can be achieved by slightly perturbing the vertices in the drawing.
    See Figure~\ref{fig:windrose-drawing}.
    \begin{figure}
        \centering
        \includegraphics{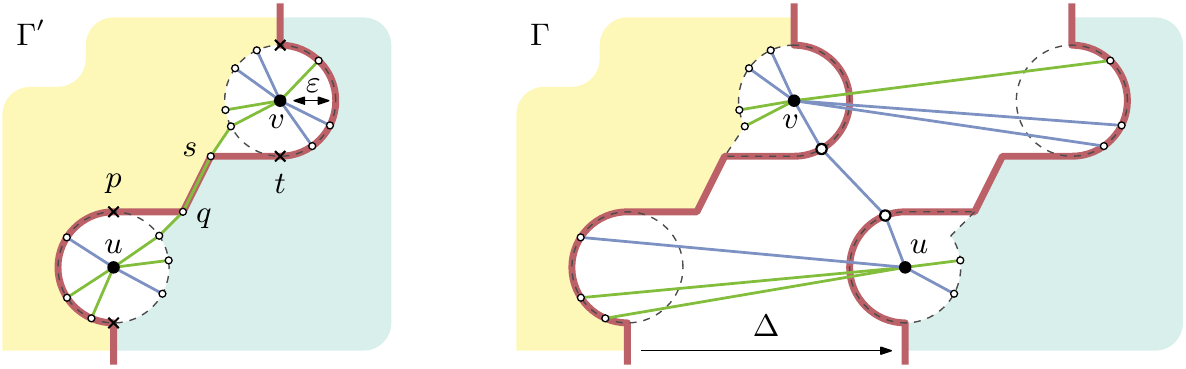}
        \caption{
            Transforming a windrose planar drawing~$\Gamma'$ of~$G'$ into a windrose planar drawing~$\Gamma$ of~$G$.
        }
        \label{fig:windrose-drawing}
    \end{figure}
    By convention, we draw NW edges in blue and NE edges in green.

    For~$\varepsilon > 0$ construct a~$y$-monotone curve~$c$ in~$\Gamma'$ as follows; see the red curve in Figure~\ref{fig:windrose-drawing}.
    The curve~$c$ extends vertically down infinitely from the point below~$u$ at distance~$\varepsilon$.
    It contains the left half-circle of radius~$\varepsilon$ centered at~$u$.
    It contains the horizontal segment from the highest point~$p$ of the half-circle up to the point~$q$ where it meets~$e$.
    Symmetrically,~$c$ extends vertically up infinitely from the point above~$v$ at distance~$\varepsilon$.
    It contains the right half-circle of radius~$\varepsilon$ centered at~$v$.
    It contains the horizontal segment from the lowest point of the half-circle~$t$ up to the point~$s$ where it meets~$e$.
    Finally,~$c$ follows~$e$ from~$q$ to~$s$.
    There exists an~$\varepsilon > 0$ so that
    (i)~$c$ intersects no vertex,
    (ii) the circles of radius~$\varepsilon$ around~$u$ and~$v$ intersect only line segments that have that vertex as an endpoint, and
    (iii) the segments from~$p$ to~$q$ and from~$s$ to~$t$ do not intersect any edge, except for~$e$.
    For property~(i), note that the vertices have been assumed to be in general position in~$\Gamma'$.
    For property~(ii), recall that edges are drawn as finite polygonal chains.
    For property~(iii), observe that all predecessors of~$e$ in~$D$ are NW edges.

    Shift the area to the right of~$c$ further right by an offset~$\Delta$ so that~$u$ lies to the right of~$v$.
    Then~$u$ and~$v$ can be connected by an NW edge.
    The segments between~$u$ and the intersection points on the circle of radius~$\varepsilon$ centered at~$u$ can be drawn as straight-line segments, preserving their monotonicity with respect to the~$x$-axis.
    The same holds true for the segments whose endpoint is~$v$.
    Finally,~$c$ may intersect edges on the vertical segments below~$u$ and above~$v$.
    Because of property~(i) shifting creates intermediate horizontal segments which can be made non-horizontal by slightly perturbing their endpoints.
\end{proof}

\section{From~$xy$-Drawings to Windrose Drawings}
\label{sec:windrose-drawings-derived-from-xy-drawings}

Let~$G = (V, E)$ be a directed graph together with sets~$X, Y$ such that~$X \cup Y = E$.
Define~$G|_Y$ as the subgraph of~$G$ induced by the edges in~$Y$.
Further, let~$\Gamma$ denote an~$xy$-drawing of~$G$.
Recall that edges are drawn as finite polygonal chains.
Define~$G^+$ as the graph obtained from~$G$ by subdividing the edges of~$G$ at each bend in~$\Gamma$ such that each directed edge~$(u, v)$ of~$G^+$ corresponds to an upward straight-line segment from~$u$ to~$v$ ($u$ is below~$v$) in~$\Gamma$.
Label~$(u, v)$ as NW (NE) if the corresponding segment decreases (increases) along the~$x$-axis.
Then~$G^+$ is a windrose graph together with a windrose planar drawing~$\Gamma^+$.
Let~$\mathcal E^+$ denote the windrose planar embedding induced by~$\Gamma^+$.

\subsection{Simplifying Windrose Planar Embeddings}

Each edge~$(u, v)$ of~$G$ in~$Y$ corresponds to a path~$(u = y_1, y_2, \dots, y_n = v)$ in~$G^+$.
For~$1 \le i < n$ the edge connecting~$y_i$ and~$y_{i + 1}$ is oriented from~$y_i$ to~$y_{i + 1}$ and is either an NE edge or an NW edge.
We can simplify~$G^+$and~$\mathcal E^+$.
Similarly-labeled edges~$(y_i, y_{i + 1}), (y_j, y_{j + 1})$ with~$1 \le i < j < n$ can be replaced by a single edge~$(y_i, y_{j + 1})$ with the same label.
We argue that the arising embedded graph is still windrose planar.
First, the order of labels around vertices~$y_i$ and~$y_{j+1}$ has not been changed, and the remaining vertices have not been affected, as a result the embedding is still windrose consistent. 
Second, we show that the embedding is still upward planar. 
This is directly implied by two facts. First, vertices~$y_{i+1}, \dots, y_{j}$ are neither sources nor sinks. And second, the sink/source assignment of~$y_{i}$ and~$y_{j+1}$ has not been changed.
Thereby the number of face-sources and face-sinks and the number of sources and sinks assigned to every face also stay the same.
Since~$\mathcal{E}^+$ is upward consistent, then the simplified embedding is upward consistent, too.
By Lemma~\ref{lem:windrose-planar-embedding-consistent}, the simplified embedding admits a windrose planar drawing (it is also an~$xy$-planar drawing of~$G$).
We can repeat this, until every edge~$(u, v) \in Y$ is represented with either a single edge, or with two edges with different labels.
In the following, we want that every gadget consists of exactly two edges.
For this purpose, if the path from~$u$ to~$v$ consists of a single edge, we subdivide it into two edges with the same label.
See Figure~\ref{fig:windrose-gadgets} for the four possible gadgets~$\mathcal H^y_1, \dots, \mathcal H^y_4$.
\begin{figure}
    \centering
    \includegraphics{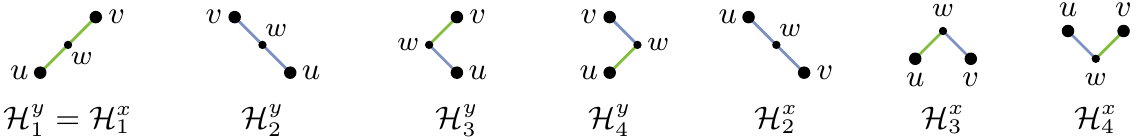}
    \caption{
        The gadgets~$\mathcal H^x_1, \dots, \mathcal H^x_4$ and~$\mathcal H^y_1, \dots, \mathcal H^y_4$ that are used to represent an~$x$-monotone edge and an~$y$-monotone edge~$(u, v)$, respectively.
        NW (NE) edges are drawn in blue (green).
    }
    \label{fig:windrose-gadgets}
\end{figure}

Each edge~$(u, v)$ of~$G$ in~$X$ corresponds to a path~$(u = x_1, x_2, \dots, x_n = v)$ in~$G^+$.
For~$1 \le i < n$ the edge connecting~$x_i$ and~$x_{i + 1}$ in~$G^+$ is either a NE edge oriented from~$x_i$ to~$x_{i + 1}$, or an NW edge oriented from~$x_{i + 1}$ to~$x_i$.
Again, we can simplify~$\mathcal E^+$.
For~$1 < i < n$ if the edges connecting~$x_{i - 1}$ and~$x_i$, and~$x_i$ and~$x_{i + 1}$ are oriented in the same direction they also are similarly labeled and the same argument as above can be used to replace~$x_i$ and its incident edges by an edge connecting~$x_{i - 1}$ and~$x_{i + 1}$.
Now consider the case that for~$1 \le i \le n - 3$ the edges~$e_1 = (x_i, x_{i + 1})$ and~$e_3 = (x_{i + 2}, x_{i + 3})$ are labeled as NE and the edge~$e_2 = (x_{i + 1}, x_{i + 2})$ is labeled as NW.
Because~$\Gamma$ is an~$xy$-drawing, the sink/source assignment induced by~$\mathcal E^+$ assigns~$x_{i + 1}$ to~$(e_2, e_1)$ and~$x_{i + 2}$ to~$(e_2, e_3)$ (in terms of upward planarity as defined in Section~\ref{sec:preliminaries}).
Replacing the vertices~$x_{i + 1}, x_{i + 2}$ and their incident edges by a single edge~$(x_i, x_{i + 3})$ labeled as NE reduces the number of face-sinks by one and it reduces the number of face-sources of both incident faces by one.
The number of assigned face-sinks and face-sources is also reduced by one.
Thereby, the sink/source assignment is still consistent and the embedding is upward planar.
Together with the arguments for edges of~$G$ in~$Y$ this shows that the simplified embedding remains a windrose planar embedding.
See Figure~\ref{fig:windrose-gadgets} for the four possible gadgets~$\mathcal H^x_1, \dots, \mathcal H^x_4$.

Let~$G^*$ and~$\mathcal E^*$ denote the simplified windrose graph and windrose planar embedding.
Every windrose planar drawing of~$G^*$ with embedding~$\mathcal E^*$ induces an~$xy$-planar drawing~$\Gamma'$ of~$G$ such that~$\Gamma$ and~$\Gamma'$ induce the same planar embedding of~$G$ and the same upward planar embedding of~$G|_Y$.
Note that~$G^*$ is obtained from~$G$ by replacing
(i) each edge of~$G$ in~$X \setminus Y$ with a gadget in~$\mathcal H^x$,
(ii) each edge of~$G$ in~$Y \setminus X$ with a gadget in~$\mathcal H^y$, and
(iii) each edge of~$G$ in~$X \cap Y$ with the (unique) gadget in~$\mathcal H^x \cap \mathcal H^y$
(where~$\mathcal H^x = \{\mathcal H_1^x, \mathcal H_2^x, \mathcal H_3^x, \mathcal H_4^x \}$ and~$\mathcal H^y = \{\mathcal H_1^y, \mathcal H_2^y, \mathcal H_3^y, \mathcal H_4^y \}$).
We say that a windrose graph obtained from~$G$ by such a gadget replacement is \emph{derived} from~$G$.
A windrose planar embedding of a graph derived from~$G$ induces a planar embedding of~$G$ and an upward planar embedding of~$G|_Y$.
We have shown the following.

\begin{lemma}
    \label{lem:xy-drawing-to-windrose-drawing-weak}
    Let~$G = (V, E)$ be a directed graph together with sets~$X, Y$ such that~$X \cup Y = E$.
    This graph admits an~$xy$-planar drawing with planar embedding~$\mathcal E$ of~$G$ and upward planar embedding~$\mathcal U$ of~$G|_Y$ if and only if
    there exists a derived graph~$G^*$ of~$G$ with a windrose planar embedding that induces~$\mathcal E$ and~$\mathcal U$.
\end{lemma}

\subsection{Special Windrose Planar Embeddings}

Inspired by Lemma~\ref{lem:xy-drawing-to-windrose-drawing-weak}, part of the approach to test~$G$ for~$xy$-planarity will be to use Lemma~\ref{lem:upward-planarity} to test for every edge~$e \in X$ and each gadget~$\mathcal H_i^x \in \mathcal H^x$ whether replacing~$e$ with~$\mathcal H_i^x$ leads to an upward planar embedding of~$G|_Y + e$.
If all edges incident to~$e$ lie in the same quadrant, it is not right away possible to derive the upward planar embedding of~$G|_Y + e$ just from the upward planar embedding of~$G|_Y$ and the gadget choice~$\mathcal H_i^x$.
For an example, consider Figure~\ref{fig:special-windrose-drawing}~(a,~b).
\begin{figure}
    \centering
    \includegraphics{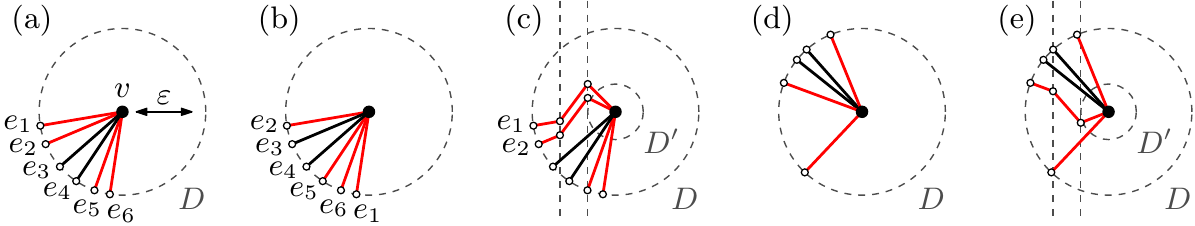}
    \caption{
        Two~$xy$-drawings that have the same cyclic order of edges around the boundary of~$D$ and assignment of line segments to quadrants, but distinct upward planar embeddings~(a,~b).
        Edges in~$X$ are drawn in red, edges in~$Y$ are drawn in black.
        The~$xy$-drawing~(a) can be locally modified to obtain a special~$xy$-drawing~(c) that admits no such ambiguities.
        Modifying a drawing where Property~(2) does not hold true works symmetrically~(d,~e).
    }
    \label{fig:special-windrose-drawing}
\end{figure}
Even though~$e_1$ might be replaced by the same gadget the sink assignment of~$G|_Y + e_1$ in~(a) is~$\psi: v \mapsto (e_4, e_1)$, whereas in~(b) it is~$\psi: v \mapsto (e_1, e_3)$.
To prevent such ambiguities, we introduce the notion of special~$xy$-drawings.

In any~$xy$-planar drawing~$\Gamma$ of~$G$ there exists some~$\varepsilon > 0$ so that the disk~$D$ of radius~$\varepsilon$ centered at~$v$
does not contain any vertex other than~$v$,
intersects only edges incident to~$v$,
and does not contain any point where an edge bends.
A counter-clockwise traversal of the boundary of~$D$ gives four (possibly empty) linear orders of the edges in each quadrant.
We say that~$\Gamma$ is \emph{special} if for each vertex~$v$ of~$G$ the following four properties hold true.
(1) If~$v$ has only incoming edges in~$X$ and incoming edges in~$Y$, then the first edge in the southwestern quadrant is in~$Y$.
(2) If~$v$ has only incoming edges in~$X$ and outgoing edges in~$Y$, then the last edge in the northwestern quadrant is in~$Y$.
(3) If~$v$ has only outgoing edges in~$X$ and incoming edges in~$Y$, then the last edge in the southeastern quadrant is in~$Y$.
(4) If~$v$ has only outgoing edges in~$X$ and outgoing edges in~$Y$, then the first edge in the northeastern quadrant is in~$Y$.

If the drawing~$\Gamma$ is not special, then we can modify it locally around each vertex where one of the four properties does not hold to obtain a special drawing.
Consider the case that~$v$ is a vertex of~$G$ that has only incoming edges in~$X$ and~$Y$ but the first edge in the southwestern quadrant is not in~$Y$, i.e., Property~(1) does not hold true for~$v$.
See Figure~\ref{fig:special-windrose-drawing}~(a).
Introduce two new bends on each edge~$e_1, e_2, \dots, e_n$ in~$X$ that precede the first edge in~$Y$ in the southwestern quadrant so that the line segments incident to~$v$ lie in the northwestern quadrant.
This preserves the~$x$-monotonicity of the drawing and ensures Property~(1).
See Figure~\ref{fig:special-windrose-drawing}~(c) for the modified drawing, where Property~(1) holds true (consider the smaller disk~$D'$).
The other three properties can be ensured symmetrically; see Figure~\ref{fig:special-windrose-drawing}~(d,~e) for an example of how to ensure Property~(2).

Note how in the special windrose drawings in Figure~\ref{fig:special-windrose-drawing}~(c,~e) for each red edge~$e$ the upward planar embedding of~$G|_Y + e$ can be derived just from the upward planar embedding of~$G|_Y$ and the gadget choice~$\mathcal H_i^x$ for~$e$.
A windrose planar embedding is \emph{special} if it is derived from a special~$xy$-planar drawing.
Observe that simplifying the windrose planar embedding as explained in the previous section does not alter edges incident to non-subdivision vertices.
Therefore, Lemma~\ref{lem:xy-drawing-to-windrose-drawing-weak} can be strengthened to special windrose planar embeddings as follows.

\begin{lemma}
    \label{lem:xy-drawing-to-windrose-drawing}
    Let~$G = (V, E)$ be a directed graph together with sets~$X, Y$ such that~$X \cup Y = E$.
    This graph admits an~$xy$-planar drawing with planar embedding~$\mathcal E$ of~$G$ and upward planar embedding~$\mathcal U$ of~$G|_Y$ if and only if
    there exists a derived graph~$G^*$ of~$G$ together with a special windrose planar embedding that induces~$\mathcal E$ and~$\mathcal U$.
\end{lemma}

Because every windrose planar graph admits a polyline drawing in polynomial area with at most one bend per edge~\cite{DBLP:journals/talg/AngeliniLBDKRR18} we immediately obtain the following.

\begin{corollary}
    Every~$xy$-planar graph admits a polyline drawing in polynomial area with at most three bends per edge.
\end{corollary}

If every embedded windrose planar graph admitted a straight-line drawing, then every~$xy$-planar graph would admit a polyline drawing with at most one bend per edge.
Not every~$xy$-planar graph admits an~$xy$-planar straight-line drawing; see Figure~\ref{fig:xy-requires-one-bend}.
\begin{figure}
    \centering
    \includegraphics{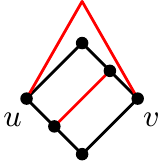}
    \caption{
        An~$xy$-planar graph that does not admit an~$xy$-planar straight-line drawing.
    }
    \label{fig:xy-requires-one-bend}
\end{figure}
Lemma~\ref{lem:xy-drawing-to-windrose-drawing} also motivates our approach of testing~$G$ for~$xy$-planarity by testing whether there exists a replacement of the edges of~$G$ with gadgets that respects the given planar and upward planar embeddings and yields a special windrose planar embedding.

\section{An~$xy$-Planarity Testing Algorithm}
\label{sec:xy-planarity-testing}

Let~$G$ be a directed graph together with edge sets~$X, Y$ and let~$e$ be an edge in~$X$.
Define the \emph{gadget candidate set}~$\mathcal H(e)$ as the subset of~$\mathcal H^x = \{\mathcal H_1^x, \mathcal H_2^x, \mathcal H_3^x, \mathcal H_4^x \}$ that contains each~$\mathcal H_i^x \in \mathcal H^x$ so that the embedding~$\mathcal U + \mathcal H_i^x$ is an upward planar embedding of~$G|_Y + e$.
Recall that Lemma~\ref{lem:xy-drawing-to-windrose-drawing} justifies that we can limit our considerations to windrose planar embeddings that are special, which lets us unambiguously derive the upward planar embedding~$\mathcal U + \mathcal H_i^x$ of~$G|_Y + e$ from the fixed upward planar embedding~$\mathcal U$ of~$G|_Y$ and the gadget choice~$\mathcal H_i^x$ for~$e$.
This is needed to test for upward planarity using Lemma~\ref{lem:upward-planarity}.
For~$e \in X$ the gadget candidate set~$\mathcal H(e)$ can be computed by tentatively replacing~$e$ with each~$\mathcal H_i^x$ and then running the upward planarity test for fixed upward embeddings.
In fact, this can even be done in overall linear time for all edges in~$X$.
To this end, choose for each face~$f$ of~$\mathcal U$ some vertex~$v$ and let~$v_1, v_2, \dots, v_k$ denote the boundary of~$f$.
Traverse the boundary of~$f$ once to compute for each~$1 \le i \le k$ the number of face-sources on the path~$v_1, v_2, \dots, v_i$, and the number of sources and sinks of~$G$ on the path~$v_1, v_2, \dots, v_i$ assigned to~$f$.
This is possible in linear time for all faces of~$\mathcal U$.
Now consider the insertion of some edge~$(v_a, v_b) \in X$ into~$f$, splitting~$f$ into two faces~$f_1, f_2$.
Using the previously calculated values, it can be checked in constant time whether~$f_1$ and~$f_2$ are upward consistent in the sense of Lemma~\ref{lem:upward-planarity}.
Together with the previous argument we conclude the following.

\begin{lemma}
    \label{lem:gadget-candidate-set-linear-time}
    The gadget candidate sets for all edges in~$X$ can be computed in linear time.
\end{lemma}

\subsection{Finding a Windrose Planar Derived Graph}

For every edge~$e$ of~$G$, add the variables~$e^{\NE}, e^{\NW}, e^{\SW}, e^{\SE}$ together with a clause~$\lnot x \lor \lnot y$ for each pair~$x, y$ of distinct variables.
This means that every edge of~$G^*$ is assigned to at most one quadrant.
Make sure that every edge of~$G^*$ is assigned to at least one quadrant as follows.
Let~$(u, v)$ be an edge of~$G$ and let~$w$ denote the vertex of the gadget in~$G^*$ that replaces~$(u, v)$.
If~$(u, v) \in X$, then add the clauses~$(u, w)^{\NE} \lor (u, w)^{\SE}$ so that~$(u, v)$ exits~$u$ in the east, and~$(w, v)^{\NW} \lor (w, v)^{\SW}$ so that~$(u, v)$ enters~$v$ from the west.
Next, if~$(u, v) \in Y$, then add the clauses~$(u, w)^{\NE} \lor (u, w)^{\NW}$ so that~$(u, v)$ exits~$u$ in the north, and~$(w, v)^{\SW} \lor (w, v)^{\SE}$ so that~$(u, v)$ enters~$v$ from the south.

Placing the edges of~$G^*$ into quadrants induces a unique gadget that replaces each edge of~$G$.
Let~$e = (u, v)$ be an edge of~$G$.
Each gadget~$\mathcal H^x_i \not\in \mathcal H(e)$ places the edge~$(u, w)$ in quadrant~$p$ and the edge~$(w, v)$ in quadrant~$q$.
Include the clause~$\lnot ( (u, w)^p \land (w, v)^q ) = \lnot (u, w)^p \lor \lnot (w, v)^q$, this would prevent the gadget~$\mathcal H^x_i$ from being induced.
Because no other gadget places the edge~$(u, w)$ in quadrant~$p$ and~$(w, v)$ in~$q$ adding this clause does not prevent any other gadget from being induced.
Add such a clause for each gadget~$\mathcal H^x_i \not\in \mathcal H(e)$.

The last step is to ensure windrose consistency at each vertex~$v$ of~$G^*$.
Use the combinatorial criterion from Lemma~\ref{lem:windrose-planar-embedding-consistent}.
Since we consider gadgets with a prescribed assignment at~$w$, we implicitly ensure that the embedding is windrose consistent around such vertices (see Figure~\ref{fig:windrose-gadgets}).
Let~$v$ be a non-subdivision vertex of~$G^*$, i.e.,~$v$ is also a vertex of~$G$.
Consider the case that~$v$ has incoming and outgoing edges in~$Y$.
Let~$\sigma = e_1, e_2, \dots, e_i, \dots, e_n, e_1$ denote the counter-clockwise cyclic order of edges incident to~$v$ in~$G^*$ such that~$e_1$ is a subdivision edge of an outgoing edge in~$Y$ and~$e_i$ is a subdivision edge of an incoming edge in~$Y$.
To achieve windrose consistency the edges must be labeled as NE, then NW, then SW and finally SE in the sequence~$e_1, \dots, e_i$.
This can be ensured with the constraints
\begin{align*}
e_j^{\NW} &\implies \lnot e_{j + 1}^{\NE} &                           &                               &                           & \\
e_j^{\SW} &\implies \lnot e_{j + 1}^{\NW} &\text{and}\quad  e_j^{\SW} &\implies \lnot e_{j + 1}^{\NE} &                           & \\
e_j^{\SE} &\implies \lnot e_{j + 1}^{\SW} &\text{and}\quad  e_j^{\SW} &\implies \lnot e_{j + 1}^{\NW} &\text{and}\quad  e_j^{\SW} &\implies \lnot e_{j + 1}^{\NE}
\end{align*}
for~$1 \le j < i$; see Figure~\ref{fig:2-sat-constraints}~(a).
\begin{figure}
    \centering
    \includegraphics{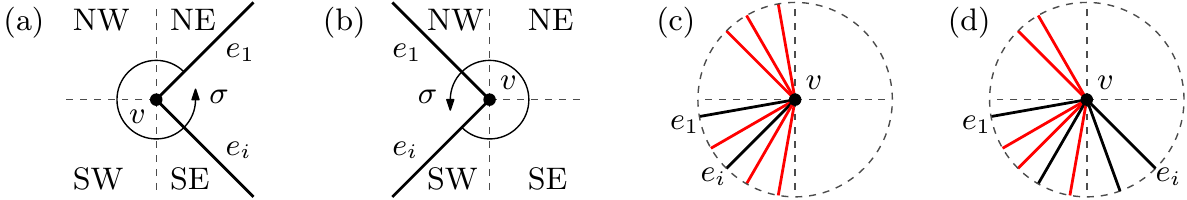}
    \caption{
        The situation around a vertex~$v$ that has both an outgoing edge~$e_1$ and an incoming edge~$e_i$ in~$Y$~(a,~b).
        An example of a vertex~$v$ that has only incoming edges in~$X$ and~$Y$~(c,~d).
        Here~$e_1$ and~$e_i$ are the leftmost and rightmost edges with respect to the fixed upward embedding.
    }
    \label{fig:2-sat-constraints}
\end{figure}
Similarly, in~$e_i, \dots, e_n, e_1$ edges must be labeled as SW, then SE, then NE and finally NW; see Figure~\ref{fig:2-sat-constraints}~(b).
A symmetric argument holds for vertices of~$G$ that have both incoming and outgoing edges in~$X$.

The remaining case consists of four subcases where~$v$ has only incoming or only outgoing edges in~$Y$, and only incoming or only outgoing edges in~$X$.
Consider the case that~$v$ has only incoming edges in~$Y$ and only incoming edges in~$X$ (the other cases are symmetric); see Figure~\ref{fig:2-sat-constraints}~(c,~d).
Let~$\sigma = e_1, e_2, \dots, e_i, \dots, e_n, e_1$ be the counter-clockwise cyclic order of edges incident to~$v$ in~$G^*$ such that
$e_1$ and~$e_i$ are subdivision edges of the leftmost and rightmost incoming edges in~$Y$ in the fixed upward planar embedding of~$G|_Y$.
Add the constraints~$e_j^{\SE} \lor e_j^{\SW}$ and~$e_j^{\SE} \implies e_{j + 1}^{\SE}$ for~$1 \le j < i$,
the constraint~$e_i^{\SE} \implies e_{i + 1}^{\NW}$,
and, because we seek special embeddings, the constraints~$e_j^{\NW} \implies e_{j + 1}^{\NW}$ for~$i < j < n$.

\subsection{Correctness}

Every windrose graph derived from~$G$ induces a solution of the \textsc{2-Sat} instance.
Every solution of the \textsc{2-Sat} instance induces a windrose graph~$G^*$ derived from~$G$ together with a windrose planar embedding.
For each edge~$e \in X$ a solution to the \textsc{2-Sat} instance induces a replacement gadget~$\mathcal H_i^x$ so that~$\mathcal U + \mathcal H_i^x$ is an upward planar embedding of~$G|_Y + e$.
The final component to our~$xy$-planarity test is to show that even though we tested the gadget candidates individually, the fact that the \textsc{2-Sat} instance is satisfiable implies that inserting for each edge in~$X$ its replacement gadget leads to an upward planar embedding of~$G^*$.

\begin{lemma}
    \label{lem:simultaneous-x-edge-insertion}
    Let~$G = (V, E)$ be a directed graph together with sets~$X, Y$ such that~$X \cup Y = E$.
    Let~$\mathcal{E}$ be an embedding of~$G$ and let~$\mathcal{U}$ be an upward planar embedding of~$G|_{Y}$.
    If the corresponding \textsc{2-Sat} instance is satisfiable, then the embedding~$\mathcal{U}^*$ of the graph~$G^*$ induced by a solution of this instance is upward planar.
\end{lemma}

\begin{proof}
    We show that~$\mathcal U^*$ is upward planar inductively by showing that inserting the gadgets one by one preserves upward planarity.
    Edges that are inserted into different faces of~$\mathcal U$ can be treated separately.
    First, consider two independent edges~$e = \{ u, v \}, e' = \{ u', v' \}$ (directed appropriately) in~$X$ that are inserted into the same face~$f$ of~$\mathcal U$; see Figure~\ref{fig:correctness}.
    \begin{figure}
        \centering
        \includegraphics{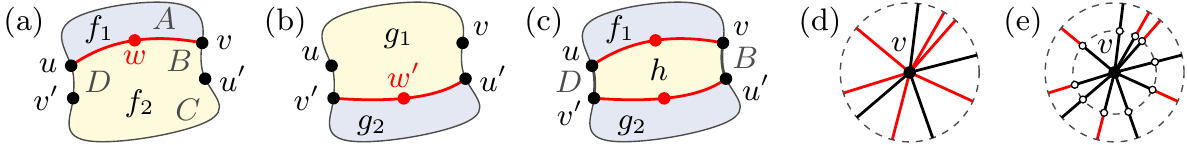}
        \caption{
            Proof of Lemma~\ref{lem:simultaneous-x-edge-insertion}.
            If~$f_1, f_2, f_1', f_2'$ are upward planar, then so is~$h$~(a--c).
            The case of adjacent edges reduces to the case of independent edges~(d,~e).
        }
        \label{fig:correctness}
    \end{figure}
    Because~$\mathcal E$ is planar the endpoints of~$e$ and~$e'$ do not alternate.
    Let~$u, A, v, B, u', C, v', D$ denote the boundary of~$f$, where~$A, B, C, D$ are sets of vertices.
    The solution of the \textsc{2-Sat} instance prescribes gadgets~$\mathcal H^x_i, \mathcal H^x_j$ that are inserted for~$e, e'$, respectively.
    Let~$w$ and~$w'$ denote the subdivision vertex of~$\mathcal H^x_i$ and~$\mathcal H^x_j$, respectively.
    Inserting~$\mathcal H^x_i$ for~$e$ splits~$f$ into two faces~$f_1, f_2$ such that the boundary of~$f_1$ is~$u, A, v, w$ and the boundary of~$f_2$ is~$u, w, v, B, u', C, v', D$; see Figure~\ref{fig:correctness}~(a).
    Similarly, inserting~$\mathcal H^x_j$ for~$e'$ splits~$f$ into two faces~$g_1, g_2$ such that the boundary of~$g_1$ is~$u', w', v', D, u, A, v, B$ and the boundary of~$g_2$ is~$u', C, v', w'$; see Figure~\ref{fig:correctness}~(b).
    Finally, inserting both~$\mathcal H^x_i$ and~$\mathcal H^x_j$ splits~$f$ into three faces, namely~$f_1$,~$g_2$ and a face~$h$ whose boundary is~$u, w, v, B, u', w', v', D$; see Figure~\ref{fig:correctness}~(c).
    From the construction of the gadget candidate sets we know that~$f, f_1, f_2, g_1, g_2$ are all upward consistent.
    We are left to show that~$h$ is upward consistent as well.
    To this end, we use Lemma~\ref{lem:upward-planarity}.

    Let~$z$ be some face of an upward embedding and let~$Z$ be a set of vertices on the boundary of~$z$.
    For the scope of this proof, let~$\psi(Z, z)$ denote the number of sources and sinks in~$Z$ assigned to~$z$.
    We have the following.
    \begin{align}
        \psi(f) &= \colorbox{green!25}{$\psi(B \cup D, f)$} + \colorbox{yellow!40}{$\psi(A, f)$} + \colorbox{red!20}{$\psi(C, f)$} + \colorbox{orange!30}{$\psi(\{ u, v \}, f)$} + \colorbox{blue!15}{$\psi(\{ u', v' \}, f)$} \label{eq:psif}\\
        \psi(f_2) &= \psi(B \cup D, f_2) + \psi(C, f_2) + \psi(\{u, v, w\}, f_2) + \psi(\{u', v'\}, f_2) \notag\\
                  &= \colorbox{green!25}{$\psi(B \cup D, f)$} + \colorbox{red!20}{$\psi(C, f)$} + \colorbox{gray!20}{$\psi(\{u, v, w\}, f_2)$} + \colorbox{blue!15}{$\psi(\{u', v'\}, f)$} \label{eq:psif2}\\
        \psi(g_1) &= \psi(B \cup D, g_1) + \psi(A, g_1) + \psi(\{u, v\}, g_1) + \psi(\{u', v', w'\}, g_1) \notag\\
                  &= \colorbox{gray!20}{$\psi(B \cup D, f)$} + \colorbox{yellow!40}{$\psi(A, f)$} + \colorbox{orange!30}{$\psi(\{u, v\}, f)$} + \colorbox{gray!20}{$\psi(\{u', v', w'\}, g_1)$} \label{eq:psig1}\\
        \psi(h) &= \psi(B \cup D, h) + \psi(\{u, v, w\}, h) + \psi(\{u', v', w'\}, h) \notag\\
                &= \colorbox{gray!20}{$\psi(B \cup D, f)$} + \colorbox{gray!20}{$\psi(\{u, v, w\}, f_2)$} + \colorbox{gray!20}{$\psi(\{u', v', w'\}, g_1)$} \label{eq:psih}
    \end{align}
    Observe that equations (\ref{eq:psif2}--\ref{eq:psih}) only hold because the upward embedding~$\mathcal U$ of~$G|_{Y}$ is fixed and the edges~$e, e'$ are independent.
    Adding~(\ref{eq:psif2}) and~(\ref{eq:psig1}) and then subtracting~(\ref{eq:psif}) gives~(\ref{eq:psih}) (gray terms remain, terms of the same color cancel out each other).
    This shows~$\psi(h) = \psi(f_2) + \psi(g_1) - \psi(f)$.

    Again, let~$z$ be some face of an upward embedding and let~$Z$ be a set of vertices on the boundary of~$z$.
    For the scope of this proof, let~$n(Z, z)$ denote the number of face-sources of~$z$ in~$Z$.
    We have the following.
    \begin{align}
        n_f &= \colorbox{green!25}{$n(B \cup D, f)$} + \colorbox{yellow!40}{$n(A, f)$} + \colorbox{red!20}{$n(C, f)$} + \colorbox{orange!30}{$n(\{ u, v \}, f)$} + \colorbox{blue!15}{$n(\{ u', v' \}, f)$} \label{eq:nf}\\
        n_{f_2} &= n(B \cup D, f_2) + n(C, f_2) + n(\{ u, v, w \}, f_2) + n(\{ u', v' \}, f_2) \notag\\
                &= \colorbox{green!25}{$n(B \cup D, f)$} + \colorbox{red!20}{$n(C, f)$} + \colorbox{gray!20}{$n(\{ u, v, w \}, f_2)$} + \colorbox{blue!15}{$n(\{ u', v' \}, f)$} \label{eq:nf2}\\
        n_{g_1} &= n(B \cup D, g_1) + n(A, g_1) + n(\{u, v \}, g_1) + n(\{u', v', w'\}, g_1) \notag\\
                &= \colorbox{gray!20}{$n(B \cup D, f)$} + \colorbox{yellow!40}{$n(A, f)$} + \colorbox{orange!30}{$n(\{u, v \}, f)$} + \colorbox{gray!20}{$n(\{u', v', w'\}, g_1)$} \label{eq:ng1}\\
        n_h &= n(B \cup D, h) + n(\{u, v, w\}, h) + n(\{u', v', w'\}, h) \notag\\
            &= \colorbox{gray!20}{$n(B \cup D, f)$} + \colorbox{gray!20}{$n(\{u, v, w\}, f_2)$} + \colorbox{gray!20}{$n(\{u', v', w'\}, g_1)$} \label{eq:nh}
    \end{align}
    Similarly, adding~(\ref{eq:nf2}) and~(\ref{eq:ng1}) and then subtracting~(\ref{eq:nf}) gives~(\ref{eq:nh}), which shows that~$n_h = n_{f_2} + n_{g_1} - n_f$.
    Let~$k = 1$ if~$f$ is the outer face and let~$k = -1$ if~$f$ is an inner face.
    Lemma~\ref{lem:upward-planarity} gives
    \begin{align}
        \psi(f)   &= n_f     + k  \label{eq:consistent-f}  \\
        \psi(f_2) &= n_{f_2} + k  \label{eq:consistent-f2} \\
        \psi(g_1) &= n_{g_1} + k. \label{eq:consistent-g1}
    \end{align}
    Adding~(\ref{eq:consistent-f2}) and~(\ref{eq:consistent-g1}) and then subtracting~(\ref{eq:consistent-f}) gives
    \[
        \psi(f_2) + \psi(g_1) - \psi(f) = n_{f_2} + n_{g_1} - n_f + k,
    \]
    which, together with~$\psi(h) = \psi(f_2) + \psi(g_1) - \psi(f)$ and~$n_h = n_{f_2} + n_{g_1} - n_f$ as shown above, gives~$\psi(h) = n_h + k$.
    This shows that~$h$ is upward consistent in the sense of Lemma~\ref{lem:upward-planarity}.

    The same idea extends to the case of faces that are bounded by gadgets corresponding to more than two independent edges.
    The case of adjacent edges reduces to the case of independent edges by subdividing for each edge~$(u, v) \in X$ the edges of the gadget that replace~$(u, v)$ and treating the subdivision edges incident to~$u$ and~$v$ as edges in~$Y$; see Figure~\ref{fig:correctness}~(d,~e).
    The choice of each gadget specifies the direction of the edges incident to~$u$ and~$v$, and, possibly, the assignment of~$u$ and~$v$ to a face.
\end{proof}

Lem.~\ref{lem:simultaneous-x-edge-insertion} gives that~$\mathcal U^*$ is upward planar.
The clauses of the \textsc{2-Sat} instance are designed such that (i)~$\mathcal U^*$ is windrose consistent by Lem.~\ref{lem:windrose-planar-embedding-consistent}, i.e.,~$\mathcal U^*$ is windrose planar, and (ii)~$\mathcal U^*$ is special.
Lem.~\ref{lem:xy-drawing-to-windrose-drawing} gives that~$G$ is~$xy$-planar if and only if the \textsc{2-Sat} instance is satisfiable.
To compute an~$xy$-planar drawing of~$G$, use the linear-time windrose-planar drawing algorithm of Angelini et al.~\cite{DBLP:journals/talg/AngeliniLBDKRR18} to compute a windrose planar drawing of~$G^*$, which induces an~$xy$-planar drawing of~$G$.

\begin{theorem}
    Let~$G = (V, E)$ be a directed graph together with subsets~$X, Y$ of its edges~$E$ such that~$X \cup Y = E$, a planar embedding~$\mathcal E$ of~$G$ and an upward planar embedding~$\mathcal U$ of~$G|_Y$.
    It can be tested in linear time whether there exists an~$xy$-planar drawing of~$G$ whose underlying planar embedding is~$\mathcal E$ and whose underlying upward planar embedding restricted to~$G|_Y$ is~$\mathcal U$.
    In the positive case, such a drawing can be computed in linear time as well.
\end{theorem}

\section{Conclusion}
\label{sec:conclusion}

We introduced and studied the concept of~$xy$-planarity which is particularly suitable to draw two posets on the same ground set, one from bottom to top and the other from left to right.
Every~$xy$-planar drawing of a graph~$G$ induces a derived windrose planar graph~$G^*$, which implies that every~$xy$-planar graph admits a polyline drawing in polynomial area with at most three bends per edge.
Because~$xy$-planarity generalizes both upward planarity and windrose planarity,~$xy$-planarity testing is \textsf{NP}-complete in general.
We considered the case that the upward part~$G|_Y$ is a connected spanning subgraph of~$G$ whose upward embedding~$\mathcal U$ is fixed, and that the planar embedding~$\mathcal E$ of~$G$ is fixed as well.
For this case, we have given a linear-time~$xy$-planarity testing algorithm.
It uses the connection to derived windrose planar graphs, a novel combinatorial view of windrose planarity and a careful analysis of upward planar embeddings and windrose planar embeddings.

\bibliographystyle{plain}
\bibliography{bibliography}

\end{document}